\documentclass[a4paper,USenglish,cleveref, autoref, thm-restate]{lipics-v2021}

\usepackage{xspace}
\usepackage[export]{adjustbox}

\hideLIPIcs  %uncomment to remove references to LIPIcs series (logo, DOI, ...), e.g. when preparing a pre-final version to be uploaded to arXiv or another public repository

\newcommand{\defparproblem}[4]{
	\vspace{3mm}
	\noindent\fbox{
		\begin{minipage}{0.96\linewidth}
			\begin{tabular*}{\linewidth}{@{\extracolsep{\fill}}lr} \textsc{#1} & {\bf{Parameter:}} #3 \\ \end{tabular*}
			{\bf{Input:}} #2 \\
			{\bf{Task:}} #4
		\end{minipage}
	}
	\vspace{2mm}
}

\newcommand{\br}[1]{\left(#1\right)}

\newcommand{\ex}[1]{\mathbb{E}\left[#1\right]}
\newcommand{\prob}[1]{\mathbb{P}\left[#1\right]}
\newcommand{\Oh}{\mathcal{O}}

\newcommand{\sub}{\subseteq}

\newcommand{\rr}{\ensuremath{\mathbb{R}}}
\newcommand{\nn}{\ensuremath{\mathbb{N}}}

\newcommand{\children}{\mathsf{Children}}
\newcommand{\fr}{\mathsf{Frac}}
\newcommand{\leaves}{\mathsf{Leaves}}
\newcommand{\schemes}{\mathsf{Schemes}}
\newcommand{\dagdp}{\textsc{Dag Disjoint Paths}\xspace}

\title{Constant Approximating Disjoint Paths on Acyclic Digraphs is W[1]-hard}
\author{Michał Włodarczyk}{University of Warsaw}{michal.wloda@gmail.com}{https://orcid.org/0000-0003-0968-8414}{Supported by Polish National Science Centre SONATA-19 grant number 2023/51/D/ST6/00155.}
\authorrunning{M. Włodarczyk} % mandatory. First: Use abbreviated first/middle names. Second (only in severe cases): Use first author plus 'et al.'

\Copyright{Michał Włodarczyk} % mandatory, please use full first names. LIPIcs license is "CC-BY";  http://creativecommons.org/licenses/by/3.0/


 % mandatory; please add comma-separated list of keywords

\category{} %optional, e.g. invited paper

\relatedversion{} %optional, e.g. full version hosted on arXiv, HAL, or other respository/website
\nolinenumbers %uncomment to disable line numbering

\EventEditors{Juli\'{a}n Mestre and Anthony Wirth}
\EventNoEds{2}
\EventLongTitle{35th International Symposium on Algorithms and Computation (ISAAC 2024)}
\EventShortTitle{ISAAC 2024}
\EventAcronym{ISAAC}
\EventYear{2024}
\EventDate{December 8--11, 2024}
\EventLocation{Sydney, Australia}
\EventLogo{}
\SeriesVolume{322}
\ArticleNo{4}
\begin{document}

\maketitle

\begin{abstract}
    In the {\sc Disjoint Paths} problem, one is given a 
    graph with a set of $k$ vertex pairs $(s_i,t_i)$ and the task is to connect each $s_i$ to $t_i$ with a path, so that the $k$ paths are pairwise disjoint.
    In the optimization variant, {\sc Max Disjoint Paths}, the goal is to maximize the number of vertex pairs to be connected.
    We study this problem on acyclic directed graphs, where {\sc Disjoint Paths} is known to be W[1]-hard when parameterized by $k$.
    We show that in this setting {\sc Max Disjoint Paths} is W[1]-hard to $c$-approximate for any constant $c$.
    To the best of our knowledge, this is the first non-trivial result regarding the parameterized approximation for {\sc Max Disjoint Paths} with respect to the natural parameter $k$.
    Our proof is based on an elementary self-reduction that is guided by a~certain combinatorial object constructed by the probabilistic method.
\end{abstract}

\section{Introduction}

The {\sc Disjoint Paths} problem has attracted a lot of attention both from the perspective of graph theory and applications~\cite{frank1990packing,ogier1993distributed,schrijver2003combinatorial,srinivas2005finding}.
Both decision variants, where one requires the paths to be either vertex-disjoint or edge-disjoint, are known to be NP-hard already on very simple graph classes~\cite{heggernes2015finding,kramer1984complexity,lynch1975equivalence,natarajan1996disjoint}.
This has motivated the study of {\sc Disjoint Paths} through the lens of parameterized complexity. Here, the aim is to develop algorithms with a running time of the form $f(k)\cdot n^{\Oh(1)}$, where $f$ is some computable function of a parameter $k$ and $n$ is the input size.
A~problem admitting such an algorithm is called {\em fixed-parameter tractable} (FPT).
In our setting, $k$ is the number of vertex pairs to be connected.
On undirected graphs, both variants of {\sc Disjoint Paths} have been classified as FPT thanks to the famous Graph Minors project by Robertson and Seymour~\cite{robertson1995graph} (see \cite{kawarabayashi2012disjoint, korhonen2024minor} for later improvements).
This was followed by a~line of research devoted to designing faster FPT algorithms on planar graphs~\cite{AdlerKKLST17, Cho0O23, LokshtanovMPSZ20, reed1995rooted, WlodarczykZ23}.

On directed graphs, there is a simple polynomial transformation between the vertex-disjoint and the edge-disjoint variants, so these two problems turn out equivalent.
Here, the problem becomes significantly harder: It is already NP-hard for $k=2$~\cite{fortune1980directed}.
The situation is slightly better for acyclic digraphs (DAGs) where {\sc Disjoint Paths} can be solved in time $n^{\Oh(k)}$~\cite{fortune1980directed} but it is W[1]-hard~\cite{Slivkins10} (cf.~\cite{amiri2016routing}) hence unlikely to be FPT.
In addition, no $n^{o(k)}$-time algorithm exists under the assumption of the Exponential Time Hypothesis (ETH)~\cite{Chitnis23}.
Very recently, it has been announced that {\sc Disjoint Paths} is FPT on Eulerian digraphs~\cite{cavallaro2024edge}.
It is also noteworthy that the vertex-disjoint and edge-disjoint variants are not equivalent on planar digraphs as the aforementioned reduction does not preserve planarity.
Indeed, here the vertex-disjoint version is FPT~\cite{DBLP:conf/focs/CyganMPP13} whereas the edge-disjoint version is W[1]-hard~\cite{Chitnis23}.

%Another perspective from which {\sc Disjoint Paths} has been studied is approximation.
In the optimization variant, called {\sc Max Disjoint Paths}, we want to maximize the number of terminals pairs connected by disjoint paths. 
The approximation status of this problem has been studied on various graph classes~\cite{chekuri2004edge, chekuri2006edge, chuzhoy2016improved, chuzhoy2017new, EneMPR16, kleinberg1995approximations, kolliopoulos2004approximating}. % , both on directed~\cite{kolliopoulos2004approximating} and undirected~\cite{chuzhoy2016improved, chuzhoy2017new} graphs.
On acyclic digraphs the best approximation factor is $\Oh(\sqrt n)$~\cite{chekuri2006} and this cannot be improved unless P=NP~\cite{chalermsook2014pre}.
A different relaxation is to allow the algorithm to output a solution in which every vertex appears in at most $c$ paths (or to conclude that there is no vertex-disjoint solution).
%The value $c$ is referred to as the congestion.
Kawarabayashi, Kobayashi, and Kreutze~\cite{kawarabayashi2014excluded} used the directed half-integral grid theorem  to design
a polynomial-time algorithm for directed {\sc Disjoint Paths} with congestion $c=4$ for every~$k$.
In other words, such a relaxed problem belongs to the class XP.
Subsequently, the congestion factor has been improved to $c=3$~\cite{kawarabayashi2015grid} and $c=2$~\cite{GiannopoulouKKK22}.

\subparagraph{Hardness of FPT approximation.}
For problems that are hard from the perspective of both approximation and FPT algorithms,
it is natural to exploit the combined power of both paradigms and consider FPT approximation algorithms.
Some prominent examples are an FPT approximation scheme for $k$-{\sc Cut}~\cite{Lokshtanov0S20} and an FPT 2-approximation for {\sc Directed Odd Cycle Transversal}~\cite{doct} parameterized by the solution size $k$. %, both running in time $f(k)\cdot n^{\Oh(1)}$.
However, several important problems proved to be resistant to FPT approximation as well.
The first hardness results in this paradigm have been obtained under a relatively strong hypothesis, called Gap-ETH~\cite{gap-eth}.
Subsequently, an $\Oh(1)$-approximation for $k$-{\sc Clique} was shown to be W[1]-hard~\cite{Lin21} and later the hardness bar was raised to $k^{o(1)}$~\cite{SK22}.
In turn, $k$-{\sc Dominating Set} is W[1]-hard to $f(k)$-approximate for any function~$f$~\cite{dominating-set} and W[2]-hard to $\Oh(1)$-approximate~\cite{LinRSW23}.
More results are discussed in the survey~\cite{feldmann2020survey}.

{\bf Proving approximation hardness under Gap-ETH is easier compared to the assumption FPT$\ne$W[1] because Gap-ETH already assumes hardness of a problem with a {\em gap}}.
Indeed, relying just on FPT$\ne$W[1] requires the reduction to perform some kind of {\em gap amplification}, alike in the PCP theorem~\cite{Dinur07}.
Very recently, the so-called {\em Parameterized Inapproximability Hypothesis} (PIH) has been proven to follow from ETH~\cite{guruswami2023parameterized}.
This means that ETH implies FPT approximation hardness of {\sc Max 2-CSP} parameterized by the number of variables within some constant approximation factor $c > 1$, which has been previously used as a starting point for parameterized reductions~\cite{bhattacharyya2021parameterized, GuruswamiRS23, lokshtanov2020parameterized, ohsaka2024parameterized}.
It remains open whether PIH can be derived from the weaker assumption FPT$\ne$W[1].

Lampis and Vasilakis~\cite{lampis2024parameterized} showed that undirected {\sc Max Vertex-Disjoint Paths} admits an FPT approximation scheme when parameterized by {\em treedepth} but, assuming PIH, this is not possible under parameterization by {\em pathwdith}.
See~\cite{ChekuriKS09, EneMPR16} for more results on approximation for {\sc Max Disjoint Paths} under structural parameterizations. 
Bentert, Fomin, and Golovach~\cite{bentert2024tight}
considered the {\sc Max Vertex-Disjoint Shortest Paths} problem where we additionaly require each path in a solution to be a shortest path between its endpoints.
They ruled out FPT($k$) approximation with factor $k^{o(1)}$ for this problem assuming FPT$\ne$W[1] and with factor $o(k)$ assuming Gap-ETH.

\subparagraph{Our contribution.}
We extend the result by Slivkins~\cite{Slivkins10} by showing that \textsc{Max Disjoint Paths} on acyclic digraphs does not admit an FPT algorithm that is a $q$-approximation, for any constant $q$.
%As far as we know, this is the first non-trivial result regarding the parameterized approximation for \textsc{Disjoint Paths}.
We formulate our hardness result as W[1]-hardness of the task of distinguishing between instances that are fully solvable from those in which less than a $\frac 1 q$-fraction of the requests can be served at once.
Since a $q$-approximation algorithm could be used to tell these two scenarios apart, the following result implies hardness of approximation.
We refer to a pair $(s_i,t_i)$ as a {\em request} that should be {\em served} by a path connecting $s_i$ to $t_i$.

\begin{restatable}{theorem}{thmMain}
\label{thm:main}
    Let $q \in \nn$ be a constant.
    It is W[1]-hard to distinguish whether for a given instance of $k$-\dagdp:
    \begin{enumerate}
        \item all the requests can be served simultaneously, or
        \item no set of $k/q$ requests can be served simultaneously.
    \end{enumerate}
\end{restatable}

%\noindent 
\vspace{0.3cm}
Our proof is elementary and
does not rely on coding theory or communication complexity as some previous W[1]-hardness of approximation proofs~\cite{dominating-set, Lin21}. 
Instead, we give a gap-amplifying self-reduction that is guided by a certain combinatorial object constructed via the probabilistic method.
%Moreover, our theorem can be extended to also the hardness can be extended

\subparagraph{Techniques.}
%The proof of \Cref{thm:main} is based on a gap-amplifying self-reduction.
A similar parameterized gap amplification technique has been previously applied to the \textsc{$k$-Steiner Orientation} problem: given a graph $G$ with both directed and undirected edges, together with a set of vertex pairs $(s_1,t_1), \dots, (s_k,t_k)$, we want to orient all the undirected edges in $G$ to maximize the number of pairs $(s_i,t_i)$ for which $t_i$ is reachable from $s_i$.
The problem is W[1]-hard and the gap amplification technique can be used to establish W[1]-hardness of constant approximation~\cite{steiner-orient}.
The idea is to create multiple copies of the original instance and connect them sequentially into many layers, in such a way that the fraction of satisfiable requests decreases as the number of layers grows.
What distinguishes \textsc{$k$-Steiner Orientation} from our setting though is that therein we do not require the $(s_i,t_i)$-paths to be disjoint.
So it is allowed to make multiple copies of each request $(s_i,t_i)$ and connect the $t_i$-vertices to the $s_i$-vertices in the next layer in one-to-many fashion.
%Repeating this idea allows to increase the gap gradually in each iteration.
Such a construction obviously cannot work for \textsc{Dag Disjoint Paths}.
Instead, will we construct a combinatorial object yielding a scheme of connections between the copies of the original instance, with just one-to-one relation between the terminals from the consecutive layers.

Imagine a following construction: given an instance $I$ of $k$-\dagdp we create $2k$ copies of $I$: $I^1_1, \dots, I^1_k$ and  $I^2_1, \dots, I^2_k$.
Next, for each $i \in [k]$ we choose some permutation $\pi_i \colon [k] \to [k]$ and for each $j \in [k]$ we connect the sink $t_j$ in $I^1_i$ to the source $s_i$ in $I^2_{\pi_i(j)}$.
See \Cref{fig:construction} on page~\pageref{fig:construction}.
Then for each $(i,j) \in [k]^2$ we request a path from the source  $s_j$ in $I^1_i$ to the sink $t_i$ in $I^2_{\pi_i(j)}$.
Observe that if $I$ is a yes-instance then we can still serve all the requests in the new instance.
\textbf{However, when $I$ is a no-instance, then there is a family $\mathcal{F}$ of $2k$ many $k$-tuples from $[k]^2$ so that each tuple represents $k$ requests that cannot be served simultaneously.}
Each tuple corresponds to some $k$ requests that have to be routed through a single copy of $I$, which is impossible when $I$ is a no-instance.

We can now iterate this argument.
In the next step we repeat this construction $k$ times (but possibly with different permutations), place such $k$ instances next to each other, and create the third layer comprising now $k^2$ copies of $I$.
Then for each $i \in [k]$ we need a permutation $\pi_i \colon [k^2] \to [k^2]$ describing the connections between the sinks from the second layer to the sources from the third layer.
Again, if $I$ is a no-instance, we obtain a family $\mathcal{F}$ of $3k^2$ many $k$-tuples from $[k]^3$
corresponding to subsets of requests that cannot be served simultaneously.
We want to show that after $d = f(k)$ many iterations no subset $A$ of 50$\%$ requests can be served.
In other words, 
the family $\mathcal{F}$ should always contain a tuple contained in $A$, certifying that $A$ is not realizable.
This will give a reduction from the exact version of \dagdp to a version of \dagdp with gap~$\frac 1 2$.
The crux of the proof is to find a collection of permutations that will guarantee the desired property~of~$\mathcal{F}$.

It is convenient to think about this construction as a game in which the first player chooses the permutations governing the connections between the layers (thus creating an instance of \dagdp) and the second player picks a subset $A$ of 50$\%$ requests.
The first player wins whenever the family $\mathcal{F}$ of forbidden $k$-tuples includes a tuple contained in $A$.
We need to show that the first player has a single winning strategy against every possible strategy of the second player.
We will prove that a good strategy for the first player is to choose every permutation independently and uniformly at random.
In fact, for a sufficiently large $d$ and any fixed strategy $A$ of the second player, the probability that $A$ wins against a randomized strategy is smaller than $2^{-k^d}$.
Since the number of possible strategies for the second player is at most $2^{k^d}$ (because there are $k^d$ requests), \textbf{the union bound implies that the first player has a positive probability of choosing a strategy that guarantees a victory against every strategy of the second player.}
This translates to the existence of a family of permutations for which the gap amplification works.

\section{Preliminaries}
We follow the convention $[n] = \{1,2,\dots,n\}$ and use the standard graph theoretic terminology from Diestel's book~\cite{diestel-book}. 
We begin by formalizing the problem.

\defparproblem{Max Disjoint Paths}{A digraph $D$, a set $\mathcal{T}$ of $k$ pairs $(s_i,t_i) \in V(D)^2$.}{k}{Find a largest collection $\mathcal{P}$ of vertex-disjoint paths so that each path $P\in\mathcal{P}$ is an $(s_i,t_i)$-path for some $(s_i,t_i) \in \mathcal{T}$.}

We refer to the pairs from $\mathcal{T}$ as {\em requests}.
A solution $\mathcal{P}$ is said to {\em serve} request $(s_i,t_i)$ if it contains an $(s_i,t_i)$-path.
The condition of vertex-disjointedness implies that each request can be served by at most one path in $\mathcal{P}$.
A {\em yes-instance} is an instance admitting a solution serving all the $k$-requests.
Otherwise we deal with a {\em no-instance}.
{(\sc Max)} \dagdp is a variant of {(\sc Max)} \textsc{Disjoint Paths} where the input digraph is assumed to be acyclic.

\subparagraph{Notation for trees.}
For a rooted tree $T$ and $v \in V(T)$ we denote by $\children(v)$ the set of direct descendants of~$v$.
A vertex $v$ in a rooted tree is a leaf if $\children(v) = \emptyset$.
We refer to 
the set of leaves of $T$ as $L(T)$. %, while $V(T)$ stands for the set of all vertices in 
The depth of a vertex $v \in V(T)$ is defined as its distance from the root, measured by the number of edges.
In particular, the depth of the root equals 0.
The set of vertices of depth $i$ in $T$ is called the $i$-th layer of $T$.

For $v \in V(T)$ we write $T^v$ to denote the subtree of $T$ rooted at $v$.
We can additionally specify an integer $\ell \ge 1$ and write $T^{v,\ell}$ for the tree comprising the first $\ell$  layers of $T^v$.
In particular, the tree $T^{v,1}$ contains only the vertex $v$. 

For $k,d \in \nn$ we denote by $T_{k,d}$ the full $k$-ary rooted tree of depth $d$.
We have $|L(T_{k,d})| = k^d$.
A~subset  $A \sub L(T_{k,d})$ is called a $q$-subset for $q\in \nn$ if $|A| \ge {|L(T_{k,d})|}\,/\,{q}$.

\subparagraph{Fixed parameter tractability.}

We provide only the necessary definitions here; more information can be found in the book~\cite{cygan2015parameterized}.
A parameterized problem can be formalized as a subset of $\Sigma^* \times \mathbb{N}$.
 We say that a problem is \emph{fixed parameter tractable} (FPT{}) if it admits an algorithm solving an
    instance \((I, k)\) %(i.e., deciding if it belongs to the language) 
    in running time
    \(f(k)\cdot |I|^{\Oh(1)}\), where \(f\) is some computable function.

    To argue that a parameterized
    problem is unlikely to be FPT{}, we employ FPT-reductions that run in time \(f(k)\cdot |I|^{\Oh(1)}\) and transform an instance $(I,k)$ into an equivalent one $(I',k')$ where $k' = g(k)$.
    A canonical parameterized problem that is believed to lie outside the class FPT is $k$-\textsc{Clique}.
    The problems that are FPT-reducible to  $k$-\textsc{Clique} form the class W[1].

\subparagraph{Negative association.}
We introduce the following concept necessary for our probabilistic argument.
There are several definitions capturing negative dependence between random variables; intuitively it means that when one variable takes a high value then a second one is more likely to take a low value.
Negative association formalizes this idea in a strong sense.

\begin{definition}
    A collection of random variables $X_1, X_2, \dots, X_n \in \rr$ is said to be {\em negatively associated} if for every pair of disjoint subsets $A_1, A_2 \sub [n]$ and every pair of increasing functions $f_1 \colon \rr^{|A_1|} \to \rr$, $f_2 \colon \rr^{|A_2|} \to \rr$ it holds that

    \[ \ex{f_1(X_i \mid i \in A_1)\cdot f_2(X_i \mid i \in A_2)} \le   \ex{f_1(X_i \mid i \in A_1)}\cdot \ex{f_2(X_i \mid i \in A_2)}. \]
\end{definition}

\noindent We make note of several important properties of negative association.

\begin{lemma}[{\cite[Prop. 3, 6, 7]{joag1983negative}}]
\label{lem:prelim:neg:properties}
    Consider a collection of random variables $X_1, X_2, \dots, X_n \in \rr$ that is {negatively associated}. Then the following properties hold.
    \begin{enumerate}
        
        \item For every family of disjoint subsets $A_1, \dots, A_k \sub [n]$ and increasing functions $f_1, \dots, f_k$, $f_i \colon \rr^{|A_i|} \to \rr$, the collection of random variables

        \[ f_1(X_i \mid i \in A_1),\, f_2(X_i \mid i \in A_2),\, \dots, f_k(X_i \mid i \in A_k)\]
        is negatively associated. \label{lem:prelim:neg:disjoint}

        \item If random variables $Y_1, \dots, Y_n$ are negatively associated and independent from $X_1, \dots, X_n$ then the collection
        $X_1, \dots, X_n, Y_1, \dots, Y_n$ is negatively associated. \label{lem:prelim:neg:union}

        \item For every sequence $(x_1, x_2, \dots, x_n)$ of real numbers we have
        \[\prob{X_i \le x_i \mid i \in [n]} \le \prod_{i=1}^n \prob{X_i \le x_i}.\]
        \label{lem:prelim:neg:orphant}
    \end{enumerate}
\end{lemma}

\begin{lemma}\label{lem:prelim:neg:sum}
    Let $n,k \in \nn$.
    For $i \in [k]$ let $\mathcal{X}^i = (X_1^i, \dots, X_n^i)$ be a sequence of real random variables that are negatively associated.
    Suppose that $\mathcal{X}^1, \dots, \mathcal{X}^k$ are independent from each other.
    Then the random variables $(\sum_{i=1}^k X^i_1,\, \dots,\, \sum_{i=1}^k X^i_n)$ are negatively associated.
\end{lemma}
\begin{proof}
    By \Cref{lem:prelim:neg:properties}(\ref{lem:prelim:neg:union}) we know that the union $\mathcal{X}^1 \cup \dots \cup \mathcal{X}^k$ forms a collection of $nk$ random variables that are negatively associated.
    We divide it into $n$ disjoint subsets of the form $(\{X_j^1, \dots, X_j^k\})_{j=1}^n$ and apply \Cref{lem:prelim:neg:properties}(\ref{lem:prelim:neg:disjoint}) for the increasing function $f \colon \rr^k \to \rr$, $f(x_1,\dots,x_k) = \sum_{i=1}^k x_i$.
\end{proof}

%\noindent 
Negative association occurs naturally in situations like random sampling without replacement.
A~scenario important for us is when an ordered sequence of numbers is being randomly permuted.
Intuitively, observing a high value at some index removes this value from the pool and decreases the chances of seeing high values at the remaining indices.

\begin{theorem}[{\cite[Thm. 2.11]{joag1983negative}}]
\label{lem:prelim:neg:permutation}
    Consider a sequence $(x_1, x_2, \dots, x_n)$ of real numbers.
    Let $\Pi \colon [n] \to [n]$ be a random variable representing a 
    permutation of the set $[n]$ chosen uniformly at random.
    For $i \in [n]$ we define a random variable $X_i = x_{\Pi^{-1}(i)}$.
    Then the random variables $X_1, X_2, \dots, X_n$ are negatively associated.
\end{theorem}
    
\section{The reduction}

\begin{figure}[t]
\centering
\includegraphics[scale=0.85]{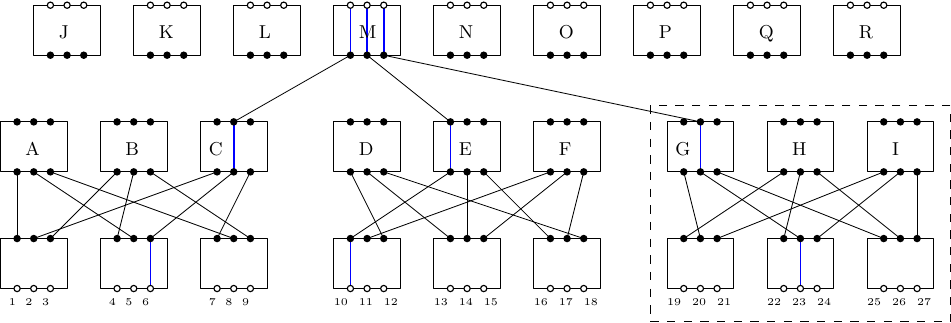}
\bigskip

\begin{adjustbox}{width=\columnwidth,center}
\begin{tabular}{|ccccccccc|ccccccccc|ccccccccc|}
\hline
J & K & L & N & O & M & P & R & Q &
M & Q & L & N & O & J & P & R & K &
P & K & L & N & M & O & J & R & Q\\
%\hline
1 & 2 & 3 & 4 & 5 & 6 & 7 & 8 & 9 & 10 & 11 & 12 & 13 & 14 & 15 & 16 & 17 & 18 & 19 & 20 & 21 & 22 & 23 & 24 & 25 & 26 & 27\\
\hline
\end{tabular}
\end{adjustbox}
\bigskip

\begin{adjustbox}{width=\columnwidth,center}
\begin{tabular}{|ccc|ccc|ccc|ccc|ccc|ccc|ccc|ccc|ccc|ccc|ccc|}
\hline
A & C & B & B & A & C & C & A & B & E & F & D & D & E & F & E & F & D & H & G & I & H & G & I & G & H & I\\
%\hline
1 & 2 & 3 & 4 & 5 & 6 & 7 & 8 & 9 & 10 & 11 & 12 & 13 & 14 & 15 & 16 & 17 & 18 & 19 & 20 & 21 & 22 & 23 & 24 & 25 & 26 & 27\\
\hline
\end{tabular}
\end{adjustbox}
\caption{An illustration for \Cref{def:game:instance} with $k=d=3$.
The boxes represent copies of an instance $I$ with $|\mathcal{T}|=3$, the large instance is $J_{3,3}(I,\beta)$ for the scheme $\beta$ listed at the bottom, and the dashed rectangle surrounds the instance $J_3 = J_{3,2}(I,\beta_3)$ where $\beta_3$ is a truncation of $\beta$ to the right subtree of $T_{3,3}$.
The hollow disks represent the sinks and sources on the large instance.
All the arcs are oriented upwards.
The leaves of $T_{3,3}$ are numbered as $1,2,\dots,27$.
For the sake of legibility, most of the arcs in the last layer are omitted and
the copies of the original instance within layers $2,3$ are marked with letters.
The letters are also used in the representation of the scheme $\beta$ which contains 9 bijections between sets of size 3 and 3 bijections between sets of size 9 (and one bijection for size 27, which is immaterial here).
The blue lines exemplify vertex pairs which belong to the request set of the large instance; the sources (in the layer 1) indexed by $6,10,23$ are mapped to the sinks in the copy $M$ (in the layer 3).
If a subset $\Gamma \sub [27]$ includes $6,10,23$ then it has a collision with respect to the scheme $\beta$.
If we work with a no-instance then such a subset $\Gamma$ of requests cannot be served as this would require routing three of them through the copy $M$.
}
\label{fig:construction}
\end{figure}

Our main objects of interest are collections of functions associated with the nodes of the full $k$-ary rooted tree.
Such a function for a node $v$ gives an ordering of leaves in the subtree of $v$.

\begin{definition}
    A {\em scheme} for $T_{k,d}$ is a collection of functions, one for each node in $T_{k,d}$, such that the function $f_v$ associated with $v \in V(T_{k,d})$ is a bijection from $L(T^v_{k,d})$ to $[|L(T^v_{k,d})|]$.
    Let $\schemes(k,d)$ denote the family of all schemes for $T_{k,d}$.
\end{definition}

We will now formalize the idea of connecting multiple copies of an instance.
On an intuitive level, we construct a $d$-layered instance %$J_{k,d}(I,\beta)$ 
by taking $k$ many $(d-1)$-layered instances and adding a new layer comprising $k^{d-1}$ copies of the original instance $I$.
Then we map the sinks in the layer $(d-1)$ to the sources in the layer $d$ according to $k$ bijections read from a scheme.
These mappings govern how we place the arcs towards the layer $d$ and which vertex pairs form the new request set.
We need a scheme $\beta \in \schemes(k,d)$ to arrange all the arcs between the layers.

In order to simplify the notation we introduce the following convention.
Suppose that an instance $J$ is being build with multiple disjoint copies of an instance  $I = (D,k,\mathcal{T})$, referred to as $I_1, I_2, \dots$.
Then we refer to the copy of the vertex $s_i \in V(D)$ (resp. $t_i$) in $I_j$ as $I_j[s_i]$ (resp. $I_j[t_i]$).

\begin{definition}\label{def:game:instance}
Given an instance $I = (D,k,\mathcal{T})$ of \dagdp and a scheme $\beta = (f_v)_{v \in V(T_{k,d})} \in \schemes(k,d)$ we construct an instance $J_{k,d}(I,\beta) = (D',k^d,\mathcal{T}')$ of \dagdp. % as follows.
The elements of $\mathcal{T}'$ will be indexed by the leaves of $T_{k,d}$
as $(s_v,t_v)_{v \in L(T_{k,d})}$ while the elements of $\mathcal{T}$ (in the instance $I$) are indexed by $1,\dots,k$ as $(s_i,t_i)_{i\in[k]}$.

If $d=1$, we simply set $J_{k,1}(I,\beta) = I$, ignoring $\beta$.
We index $\mathcal{T}$ by $L(T_{k,1})$ in an arbitrary order.

Consider $d > 1$. Let $r$ be the root of $T_{k,d}$ with $\children(r) = \{u_1, \dots, u_k\}$.
For $i \in [k]$ let $\beta_i$ be the truncation of $\beta$ to the nodes in the subtree  $T^{u_i}_{k,d}$ and $J_i = (D_i, k^{d-1}, \mathcal{T}_i)$ be the instance $J_{k,d-1}(I, \beta_i)$.
We take a disjoint union of $J_1,\dots,J_k$ and $k^{d-1}$ copies of $I$ referred to as $I_1,I_2,\dots$ (see \Cref{fig:construction}).
These  $k^{d-1}$ copies of $I$ form layer $d$.

%For a leaf $v \in L(T_{k,d})$ we define $i(v) \in [k]$ as the index for which $u_{i(v)}$ is an ancestor of $v$. Then $v$ belongs to $L(T^{u_{i(v)}}_{k,d})$ and the value $f_{u_{i(v)}}(v)$ is well-defined.

Recall that for $i \in [k]$ the bijection $f_{u_i}$ maps $L(T^{u_i}_{k,d})$ to $[k^{d-1}]$.
For each $i \in [k]$ and  $v \in L(T^{u_i}_{k,d})$ we insert an arc from $J_{i}[t_v]$ to $I_{f_{u_i}(v)}[s_i]$.
%For each  $v \in L(T_{k,d})$ and $i = i(v)$ we insert an arc from $J_{i}[t_v]$ to $I_{f_{u_i}(v)}[s_i]$.
%Finally, for each $i \in [k], j \in L(T^{u_i}_{k,d})$ we insert an arc from $D_i[t_j]$ to $I_{f_{u_i}(j)}[s_i]$.
Then we add the pair $(J_{i}[s_v], I_{f_{u_i}(v)}[t_i])$ to $\mathcal{T}'$.
This pair is assigned index $\iota(v)$ in $\mathcal{T}'$ where $\iota$ is the natural embedding $L(T^{u_i}_{k,d}) \to L(T_{k,d})$.
%We define the terminal set $\mathcal{T}'$ as $\{(D_i[s_v], I_{f_{u_i}(j)}[t_i]) \mid i \in [k], j \in [k^{d-1}]\}$. \micr{make a mapping to $A$}
\end{definition}

Note that whenever $D$ is acyclic then $D'$ is acyclic as well so the procedure indeed outputs an instance of \dagdp.
It is also clear that when $I$ admits a solution serving all the $k$ requests, it can be used to serve all the requests in $J_{k,d}(I,\beta)$.

\begin{observation}\label{lem:game:yes}
    Let $k,d \in \nn$ and $\beta \in \schemes(k,d)$.
    If $I = (D,k,\mathcal{T})$ is a yes-instance of \dagdp
    then $J_{k,d}(I,\beta)$ is a yes-instance as well.
\end{observation}

The case when $I$ is a no-instance requires a more careful analysis.
We introduce the notion  of a {\em collision} that certifies that some subset of requests cannot be served.

\begin{definition}
    Let $k,d \in \nn$, $A \sub L(T_{k,d})$, and $\beta = (f_v)_{v \in V(T_{k,d})} \in \schemes(k,d)$.
    We say that $u \in V(T_{k,d})$ forms a {\em collision} with respect to $(A,\beta)$ if $A$ contains elements $a_1,\dots,a_k$
such that:
\begin{enumerate}
    \item for each $i \in [k]$ the node $a_i$ is a descendant of $u_i \in \children(u)$ where $u_1,\dots,u_k$ are distinct,
    \item $f_{u_1}(a_1) = f_{u_2}(a_2) = \dots = f_{u_k}(a_k)$.
\end{enumerate}
\end{definition}

\begin{lemma}\label{lem:game:no}
    Let $k,d \in \nn$, $A \sub L(T_{k,d})$, and $\beta = (f_v)_{v \in V(T_{k,d})} \in \schemes(k,d)$.
    Suppose that there exists a collision with respect to $(A,\beta)$.
    Let $I = (D,k,\mathcal{T})$ be a no-instance of \dagdp.
    Then no solution to the instance $(D',k^d,\mathcal{T}') = J_{k,d}(I,\beta)$ can simultaneously serve all the requests $\{(s_v,t_v)_{v \in A}\}$.
\end{lemma}
\begin{proof}
    We will prove the lemma by induction on $d$.
    In the case $d=1$ we have $J_{k,1}(I,\beta) = I$ and the only possibility of a collision is when $A = L(T_{k,1})$ so $\{(s_v,t_v)_{v \in A}\}$ is the set of all the requests. 
    By definition, we cannot serve all the requests in a no-instance.
    Let us assume $d > 1$ from now on.
    
    First suppose that the collision occurs at the root $r \in V(T_{k,d})$.
    Let $\children(r) = \{u_1, \dots, u_k\}$.
    Then there exists $A' = \{a_1,\dots,a_k\} \sub A$ such that $a_i$ is a descendant of $u_i$ and $f_{u_1}(a_1) = f_{u_2}(a_2) = \dots = f_{u_k}(a_k)$.
    We refer to this common value as $x = f_{u_i}(a_i)$.
    We will also utilize the notation from \Cref{def:game:instance}.

    Observe that in order to serve the request $(s_{a_i},t_{a_i})$ in $D'$ the path $P_i$ starting at $s_{a_i} = J_i[s_{a_i}]$ must traverse the arc from $J_i[t_{a_i}]$ to $I_x[s_i]$ as every 
    other arc leaving $D_i$ leads to some $I_y$ with $y \ne x$ having no connection to $t_{a_i} = I_x[t_i]$.
    Furthermore, the path $P_i$ must contain a subpath connecting $I_x[s_i]$ to $I_x[t_i]$ in $I_x$.
    Since the same argument applies to every $i \in [k]$, we would have to serve all the $k$ requests in $I_x$.
    But this is impossible because $I_x$ is a copy of $I$ which is a no-instance. 

    Now suppose that the collision does not occur at the root.
    Then it must occur in the subtree $T^{u_i}_{k,d}$
    for some $i \in [k]$.
    For every $v \in A$ being a descendant of $u_i$,
    any path $P_v$ serving the request $(s_{v},t_{v})$ in $D'$ must contain a subpath $P'_v$ in $D_i$ from $J_i[s_{v}]$ to $J_i[t_{v}]$ as again it must leave $D_i$ through the vertex $J_i[t_{v}]$.
    By the inductive assumption, we know that we cannot simultaneously serve all the requests $(s_v,t_v)_{v \in A \cap L(T^{u_i}_{k,d})}$ in the smaller instance $J_i$.
    The lemma follows.

    \iffalse
    Consider some pair $(s_i, t_i) \in \mathcal{T}'$.
    We can associate with it a collection of $d$ triples $(I_j,x_j)$, where $I_j$ is a copy of $I$ inside $J_{k,d}(I,\beta)$ and $x_j \in [k]$. 
    
    %$(x_j,y_j)$ is a source-sink pair in $D_j$.
    If $d=1$ the collection contains just the triple $(D, s_i, t_i)$.
    If $d > 1$ we consider the copy of $I$ containing $t_i$ and start with the pair $(x,t_i)$ where $x$ is the source corresponding to the sink 

    Let $u \in V(T_{k,d}$ be the node forming the collision,
    with $\children(u) = \{u_1, \dots, u_k\}$.
    We know that there exist $A' = \{a_1,\dots,a_k\} \sub A$ such that $a_i$ is a descendant of $a_i$ and $f_{u_1}(a_1) = f_{u_2}(a_2) = \dots = f_{u_k}(a_k)$.
    We will refer to the latter number as $x$.
    
    Let us focus on the moment in the construction process of $J_{k,d}(I,\beta)$ when $u$ is considered.
    We have just created instances $J_1, \dots, J_k$ corresponding to the subtrees of $u_1, \dots, u_k$.
    As the next step, we insert $k^{d-1}$ copies of $I$.
    For each $i \in [k]$ the vertex $D_i[t_{a_i}]$ gets connected to $I_x[s_i]$.
    \fi 
\end{proof}

\iffalse
Alice and Bob play a game on $T_{k,d}$.
The strategy of Alice is a subset $A \sub L(T_{k,d})$.
We say that Alice plays a {\em $q$-strategy} if $|A| \ge \frac{|L(T_{k,d})}{q} = \frac{k^d}{q}$.
The strategy of Bob is a scheme $\beta \in \schemes(k,d)$.
Bob wins if there exists a collision with respect to $(A,\beta)$.
\fi

We can now state our main technical theorem.
Recall that a subset  $A \sub L(T_{k,d})$ is called a $q$-subset if $|A| \ge {|L(T_{k,d})|}/{q} = {k^d}/{q}$.

\begin{restatable}{theorem}{thmBob}
\label{thm:game:bob}
    Let $k,d,q \in \nn$ satisfy  $d \ge k \cdot (4q)^{4k\log k}$.
    Then there exists $\beta \in \schemes(k,d)$
    such that for every $q$-subset $A \sub L(T_{k,d})$ there is a collision with respect to $(A,\beta)$. 
\end{restatable}

The proof is postponed to \Cref{sec:scheme} which abstracts from the \textsc{Disjoint Paths} problem and focuses on random permutations.
With \Cref{thm:game:bob} at hand, the proof of the main result is easy.

\thmMain*
\begin{proof}
    We are going to give an FPT-reduction from the exact variant of $k$-\dagdp, which is W[1]-hard~\cite{Slivkins10}, to the variant with a sufficiently large gap.
    To this end, we present an algorithm that, given an instance $I = (D,k,\mathcal{T})$,
    runs in time $f(k,q) \cdot |I|$ and outputs an instance $J = (D',k',\mathcal{T}')$ such that:
    \begin{enumerate}
        \item $k'$ depends only on $k$ and $q$,
        \item if $I$ is a yes-instance then $J$ is a yes-instance, and
        \item if $I$ is a no-instance then no solution to $J$ can simultaneously serve at least $k'/q$ requests.
    \end{enumerate}
    Obviously, being able to separate these two cases for $J$ (all requests vs. at most $\frac 1 q$-fraction of requests) is sufficient to determine whether $I$ is a yes-instance.
    
    We set $d = k \cdot (4q)^{4k\log(k)}$ accordingly to \Cref{thm:game:bob}.
    It guarantees that there exists a scheme $\beta \in \schemes(k,d)$
    such that for every $q$-subset $A \sub L(T_{k,d})$ there is a collision with respect to $(A,\beta)$.
    Observe that such a scheme can be computed in time $f(k,q)$ because $d$ is a function of $(k,q)$ and the size of the family $\schemes(k,d)$ is a function of $(k,d)$.
    The same holds for the number of all $q$-subsets $A \sub L(T_{k,d})$.
    Therefore, we can simply iterate over all $\beta \in \schemes(k,d)$ and check for each $q$-subset $A$ whether there is a collision or not.

   The instance $J$ is defined as $J_{k,d}(I,\beta)$. A direct implementation of \Cref{def:game:instance} takes time $f(k,d) \cdot |I|$.
    \Cref{lem:game:yes} says that if $I$ is a yes-instance, then $J$ is as well, whereas \Cref{lem:game:no} ensures that if $I$ is a no-instance, then for each set of $k'/q$ requests (corresponding to some $q$-subset $A \sub L(T_{k,d})$ which must have a collision with $\beta$) no solution can simultaneously serve all of them.
    This concludes the correctness proof of the reduction.
\end{proof}

We remark that \Cref{thm:main} works in a more general setting, where $q$ is not necessarily a constant, but a function of $k$.
This enables us to rule out not only an $\Oh(1)$-approximation in FPT time, but also an $\alpha(k)$-approximation for some slowly growing function $\alpha(k) \to \infty$.
However, the value of the parameter $k'$ becomes $k^d$ for $d = \Omega(q^{k\log k})$
so $q$ ends up very small compared to the new parameter~$k'$.
This is only sufficient to rule out approximation factors of the form $\alpha(k) = (\log k)^{o(1)}$.
A detailed analysis of how to adjust such parameters is performed in~\cite{steiner-orient}.

\section{Constructing the scheme}
\label{sec:scheme}

This section is devoted to the proof of \Cref{thm:game:bob}.
Before delving into the rigorous analysis, we sketch the main ideas behind the proof.

\subparagraph{Outline.}
We use the probabilistic method to prove the existence of a scheme having a collision with every $q$-subset of leaves in $T_{k,d}$.
We will show that for a sufficiently large $d$ choosing each bijection at random yields a very high probability of a collision with any fixed $q$-subset.
Specifically, the probability that a collision does not occur should be less than $2^{-k^d}$.
Since the number of all $q$-subsets of a $k^d$-size set is bounded by $2^{k^d}$, the union bound will imply that the probability that a collision does not occur for at least one $q$-subset is strictly less than one, implying the existence of the desired scheme.

Let us fix a $q$-subset $A \sub L(T_{k,d})$.
Suppose there is a vertex $u \in V(T_{k,d})$ such that for every child $y$ of $u$ the fraction of leaves in $T^y_{k,d}$ belonging to $A$ is at least $1/q$.
Let $\ell$ denote $[|L(T^y_{k,d})|]$.
For each such child we choose a random bijection from $L(T^y_{k,d})$ to $[\ell]$.
The probability that each of these $k$ bijections maps an element of $A$ to a fixed index $x \in [\ell]$ is at least $q^{-k}$.
Such events are not independent for distinct $x$ but we will see that they are negatively associated, which still allows us to upper bound the probability of no such event happening by $(1-q^{-k})^\ell$ (see \Cref{lem:scheme:assosiation}).

How to identify such a vertex $u$?
First, it is sufficient for us to relax the bound $1/q$ assumed above to $1/(4q)$. 
Observe that for each layer in   $T_{k,d}$ there must be many vertices $v$ satisfying $|A \cap L(T^v_{k,d})| \ge \frac 1 {2q} |L(T^v_{k,d})|$.
Suppose that $v$ does not meet our criterion: this means that it has a child $v'$ with less than $1/(4q)$-fraction of the $A$-leaves in its subtree.
But then the average fraction of the $A$-leaves among the remaining children is higher than the fraction for $v$.
Consequently, we can choose a child of $v$ with a higher fraction and repeat this argument inductively.
We show that after $\Oh(k\log(q))$ many steps this process must terminate so we are guaranteed to find a vertex for which every child has at least a $1/(4q)$-fraction of the $A$-leaves. This is proven in \Cref{lem:tree-one}.

Finally, to obtain a large probability of a collision we must show that there many such vertices $u$ with a large sum of their subtrees' sizes.
This will allows us to multiply the aforementioned bounds of the form $(1-q^{-k})^\ell$ with a large sum of the exponents $\ell$.
By applying the argument above to a single layer in $T_{k,d}$ we can find such a collection with the sum of their subtrees' sizes being $k^d$ divided by some function of $k$ and $q$.
But we can also apply it to multiple layers as long as they are sufficiently far from each other (so that the vertices found by the inductive procedure are all distinct). 
Therefore, it suffices to take $d$ large enough so that the number of available layers surpasses the factors in the denominator, which depend only on $k$ and $q$.
This is analyzed in \Cref{lem:scheme:F}.

\subparagraph{} %{Ingredients for the proof.}
We begin with a probabilistic lemma stating that randomly permuting $k$ large subsets of a common universe yields a large chance of creating a non-empty intersection of these sets.

\begin{lemma}\label{lem:scheme:assosiation}
    Let $k, z, \ell \in \nn$ and
    $X_1, \dots, X_k$ be subsets of $[\ell]$ of size at least $\ell / z$ each.
    Next, let $\Pi_1, \dots, \Pi_k \colon [\ell] \to [\ell]$ be independent random variables with a uniform
distribution on the family of all permutations over the set $[\ell]$. Then
   % with probability  the intersection $X_1\cap X_2 \cap \dots X_k$ is non-empty. %there exists $i \in [\ell]$ 
    \[
     \prob{\Pi_1(X_1) \cap \Pi_2(X_2) \cap \dots \cap \Pi_k(X_k) = \emptyset} \le \exp(-\ell / z^{k}).
    \]
\end{lemma}
\begin{proof}
    For $i \in [k]$ and $j \in [\ell]$ let $Y^i_j = 1$
    if $j \in \Pi_i(X_i)$ and $Y^i_j = 0$ otherwise.
    By \Cref{lem:prelim:neg:permutation} the variables $(Y^i_1, \dots, Y^i_\ell)$ have negative association for each $i \in [k]$.
    Note that $\mathbb{E} {Y^i_j} \ge 1/z$.
    %Then for each $i \in [k]$ the variables $Y^i_1, \dots, Y^i_\ell$ are negatively correlated.
    Next, let $Z_j = \sum_{i=1}^k Y^i_j$ for $j \in [\ell]$.
    \Cref{lem:prelim:neg:sum} ensures that the variables $Z_1, \dots, Z_\ell$ also enjoy negative association.
    Condition $\Pi_1(X_1) \cap \Pi_2(X_2) \cap \dots \cap \Pi_k(X_k) = \emptyset$ is equivalent to $\max(Z_j)_{j=1}^\ell \le k - 1$.
    We have
    \[\prob{Z_j \le k-1} = 1 - \prob{Z_j = k} = 1 - \prod_{i=1}^k\prob{j \in \Pi_i(X_i)} \le 1 - 1/z^k.\]
    \[\prob{\max(Z_j)_{j=1}^\ell \le k-1} \le \prod_{j=1}^\ell\prob{Z_j \le k-1} \le (1 - 1/z^{k})^\ell = (1 - 1/z^{k})^{z^{k} \cdot (\ell / z^{k})} \le  \exp(-\ell / z^{k}).\]
    In the first inequality we used \Cref{lem:prelim:neg:properties}(\ref{lem:prelim:neg:orphant}).
    The last one holds because $(1-\frac 1 m)^m < \frac 1 e$ for all $m \ge 2$.
\end{proof}

\subparagraph{Notation.}
We introduce some additional notation to work with the tree $T_{k,d}$.
For a vertex $v \in V(T_{k,d})$ let $\leaves(v)$ denote the size of the set $L(T^v_{k,d})$.
Note that $\leaves(v) = k^{d-h}$ where $h$ is the depth of~$v$.
Next, for a set $A \sub L(T_{k,d})$, we will write $\fr_A(v) = |A \cap L(T^v_{k,d})|\, /\, \leaves(v)$.
When $A$ is clear from the context, we will omit the subscript.

\begin{lemma}\label{lem:tree-one}
    Let $k, d, q, \tau \in \nn$ satisfy $k,q\ge 2$ and $d \ge \tau \ge 2k \cdot \log(q)$.
    Next, let $v \in V(T_{k,d})$ be of depth at most $d-\tau$ and $A \sub L(T_{k,d})$  satisfy $\fr_A(v) \ge \frac 1 q$. %|A \cap L(T^v_{k,d})| \ge \frac 1 q \cdot |L(T^v_{k,d})|$.
    % Let $A \sub L(T_{k,d})$ be of size at least $\frac 1 q \cdot |L(T_{k,d})| = \frac{k^d}{q}$.\micr{q strategy}
    Then there exists a vertex $u \in V(T^{v,\tau}_{k,d})$ such that for each $y \in \children(u)$ it holds that $\fr_A(y) \ge \frac 1 {2q}$. %$|A \cap L(T^u_{k,d})| \ge \frac 1 {2q} \cdot |L(T^u_{k,d})|$.
\end{lemma}
\begin{proof}
    %For $u \in V(T_{k,d})$ we will write $\fr(u) = |A \cap L(T^u_{k,d})|\, /\, |L(T^u_{k,d})|$.
    %We have $\fr(v) \ge \frac 1 q$ and we want to show that there exists a node $u$ in $T^{v,\tau}_{k,d}$ with $\fr(u') \ge \frac 1 {2q}$ for all $u' \in \children(u)$.
    Suppose the claim does not hold.
    We will show that under this assumption for each $i \in [\tau]$ there exists $v_i \in V(T^{v,i}_{k,d})$ with $\fr(v_i) \ge \frac 1 q \cdot (1 + \frac 1 {2k-2})^{i-1}$.
    Then by substituting $i = \tau > (2k-2) \cdot \log(q)$ and estimating $(1 + \frac 1 m)^m >2$ (for all $m \ge 2$) we will arrive at a contradiction: \[ \fr(v_{\tau}) \ge \frac 1 q \cdot \br{1 + \frac 1 {2k-2}}^{(2k-2) \cdot \log(q)} > 
    \frac 1 q \cdot 2^{\log(q)} \ge 1.\]

    We now construct the promised sequence $(v_i)$ inductively.
    For $i = 1$ we set $v_1 = v$ which obviously belongs to $T^{v,1}_{k,d}$ and satisfies $\fr(v) \ge \frac 1 q$.
    To identify $v_{i+1}$ we consider $\children(v_i) = u_1, u_2, \dots, u_k$.
    We have $\fr(v_i) = \frac 1 k \cdot \sum_{j=1}^k \fr(u_j)$.
    We define $v_{i+1}$ as the child of $v_{i}$ that maximizes the value of $\fr$ (see \Cref{fig:tree}).
    By the assumption, one of the children satisfies $\fr(u_j) < \frac 1 {2q}$.
    Then $\fr(v_{i+1})$ is lower bounded by the average value of $\fr$ among the remaining $k-1$ children, which is at least
    $\frac 1 {k-1} \br{\fr(v_i)\cdot k - \frac 1 {2q}}$.
    We have $\fr(v_i) \ge \fr(v_1) \ge \frac 1 q$ so $\br{\fr(v_i)\cdot k - \frac 1 {2q}} \ge \br{\fr(v_i)\cdot k - \frac {\fr(v_i)} {2}}$.
    We check that $v_{i+1}$ meets the specification:
    \[\fr(v_{i+1}) \ge \frac{\fr(v_i)}{k-1}\cdot \br{k - \frac 1 2}  = \fr(v_i)\cdot \br{1 + \frac{1}{2k-2}} \ge \frac 1 q \cdot \br{1 + \frac 1 {2k-2}}^{i}\]
    In the last inequality we have plugged in the inductive assumption. The lemma follows.
\end{proof}

To apply \Cref{lem:tree-one} we need to identify many vertices satisfying $\fr_A(v) \ge \frac 1 {2q}$.
To this end, we will utilize the following simple fact.

\begin{lemma}\label{lem:mean}
    Let $a_1, a_2, \dots, a_\ell \in [0,1]$ be a sequence with mean at least $x$ for some $x \in [0,1]$.
    Then at least $\frac{x\ell}{2}$ elements in the sequence are lower bounded by $\frac x {2}$.
\end{lemma}
\begin{proof}
    Suppose that $|\{a_i \ge \frac x {2} \mid i \in [\ell]\}| < \frac{x\ell}{2}$.
    This leads to a contradiction: 
    \[\sum_{i=1}^\ell a_i < 1\cdot \frac{x\ell}{2} + \frac x {2} \cdot \ell %\br{1-\frac x 2} <
    %\frac{x\ell}{2} + \frac{x\ell}{2}
    = {x\ell}.\]
\end{proof}

%The bound above is not tight but it will be sufficient for our needs.

We will use the lemmas above for a fixed layer in the tree $T_{k,d}$ to identify  multiple vertices $v$ meeting the requirements of \Cref{lem:tree-one}.
For each such $v$ we can find a close descendant $u$ of $v$ for which we are likely to observe a collision. %, due to \Cref{lem:scheme:assosiation}.
The value $\ell$ in \Cref{lem:scheme:assosiation}, governing the probability of a collision, corresponds to the number of leaves in the subtree of $u$, i.e., $\leaves(u)$.
Since this value appears in the exponent of the formula, we need a collection of such vertices $u$ in which the total sum of $\leaves(u)$ is large.

\begin{figure}[t]
\centering
\includegraphics[scale=1.05]{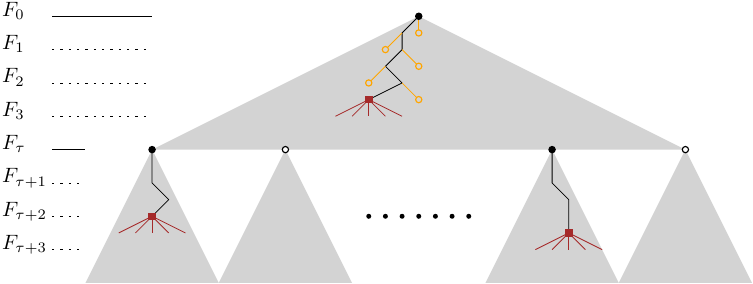}
\caption{An illustration for \Cref{lem:scheme:F}.
We consider layers $F_0, F_{\tau}, F_{2\tau},\, \dots$
The vertices from $F_0^+$ and $F_\tau^+$ are marked by black disks and their subtrees $F^{v,\tau}_{k,d}$ are depicted as gray triangles. 
For each vertex $v \in F^+$ we apply \Cref{lem:tree-one} to identify a vertex $\gamma(v) \in F$: the red square inside the corresponding triangle.
The root also illustrates the argument from \Cref{lem:tree-one}. 
We start with a vertex $v$ satisfying $\fr(v) \ge \frac 1 {2q}$ and while one of its children $v'$ has $\fr(v') < \frac 1 {4q}$ we can find another child $v''$ of $v$ with $\fr(v'') > \fr(v)$.
This process terminates within $\tau$ steps.
}\label{fig:tree}
\end{figure}

\begin{lemma}\label{lem:scheme:F}
    Let $k, d, q \in \nn$ satisfy $k,q \ge 2$, $d \ge 4kq$. % satisfy $d \ge 2k \cdot \log(q)$.
    If $A \sub L(T_{k,d})$ is a $q$-subset % be of size at least $\frac 1 p \cdot |L(T_{k,d})| = \frac{k^d}{p}$.
    then there exists a set $F \sub V(T_{k,d})$ with the following properties.
    \begin{enumerate}
        \item For each $v \in F$ and $u \in \children(v)$ it holds that $\fr_A(u) \ge  \frac 1 {4q}$.
        \item The sum $\sum_{v\in F} \leaves(v)$ equals at least  ${d \cdot k^{d}} \cdot (4q)^{-3k \log(k)}$.
    \end{enumerate}
\end{lemma}
\begin{proof}
    %As before, we  write $\fr(u) = |A \cap L(T^u_{k,d})|\, /\, |L(T^u_{k,d})|$.
    Let $F_i \sub V(T_{k,d})$ be $i$-th layer of $T_{k,d}$, i.e., the set of vertices of depth $i$; we have $|F_i| = k^{i}$ and $\leaves(v) = k^{d-i}$ for each $v \in F_i$.
    Since their subtrees are disjoint, we can see that %$\sum_{v\in F_i} \leaves(v) = |L(T_{k,d})| = k^d.$ Similarly 
    $\sum_{v\in F_i} |A \cap L(T^v_{k,d})| = |A|$.
    Therefore $\sum_{v\in F_i} \fr(v) / |F_i| \ge \frac 1 q$.
    By \cref{lem:mean} at least $\frac 1 {2q}$ fraction  of the vertices in $F_i$ must satisfy $\fr(v) \ge \frac 1 {2q}$.
    Let us denote this subset as $F^+_i$.

    Let $\tau = \lceil 2k \cdot \log (2q) \rceil$ and $M = \lfloor d / \tau \rfloor$. %\micr{some assumption on $d$ to make this positive}
    %We have $M \tau < d$.
    We define $F^+ = F^+_0 \cup F^+_\tau \cup F^+_{2\tau} \cup \dots \cup F^+_{(M-1)\tau}$.
    Observe that for each pair $u, v \in F^+$ the trees $T^{u,\tau}_{k,d},\, T^{v,\tau}_{k,d}$ are disjoint.
    We apply \cref{lem:tree-one} with $q' = 2q$ to each $v \in F^+$ to obtain a vertex $\gamma(v) \in V(T^{v,\tau}_{k,d})$ satisfying condition (1).
    The disjointedness of these subtrees ensures that the vertices $\gamma(v)_{v \in F^+}$ are distinct.
    We define $F = \{\gamma(v) \mid v \in F^+\}$.

    Now we take care of condition (2).
    Let us fix $j \in [0, M-1]$.
    Since $\gamma(v) \in V(T^{v,\tau}_{k,d})$ for $v$ with the depth $j\tau$, we infer that the depth of $\gamma(v)$ is at most $(j+1)\tau-1$ so $|\leaves(\gamma(v))| \ge k^{d+1-(j+1)\tau}$.
    We have established already that $|F^+_{j\tau}| \ge \frac{|F_{j\tau}|}{2q} =  \frac{k^{j\tau}} {2q}$.
    The assumption $d \ge 4kq$ implies $d \ge \tau$ so we can simplify $M = \lfloor d / \tau \rfloor \ge d / (2\tau)$.
    We estimate the sum within each layer $F^+_{j\tau}$ and then multiply it by $M$.
    
    \[ \sum_{v\in F^+_{j\tau}} \leaves(\gamma(v)) \ge \frac{k^{j\tau}} {2q} \cdot k^{d+1-(j+1)\tau} = \frac{k^{d + 1 - \tau}} {2q}\]
    \[ \sum_{v\in F} \leaves(v) = \sum_{j = 0}^{M-1} \sum_{v\in F^+_{j\tau}} \leaves(\gamma(v)) \ge \frac{d \cdot k^{d + 1 - \tau}} {2\tau \cdot 2q} \]
    To get rid of the ceiling, we estimate $\tau \le 2k \cdot \log (4q)$. 
    Then $k^\tau \le k^{2k \log (4q)} = (4q)^{2k \log(k)}$.
    We also use a trivial bound $\tau \le 4kq$. 
    We can summarize the analysis by

    \[\sum_{v\in F} \leaves(v) \ge \frac{d \cdot k^{d + 1 - \tau}} {2\tau \cdot 2q} = \frac{d \cdot k^{d+1}} {k^{\tau} \cdot 4q\tau} \ge 
    \frac{d \cdot k^{d+1}} {(4q)^{2k \log(k)} \cdot 16kq^2} \ge \frac{d \cdot k^{d}} {(4q)^{3k \log(k)} } \] 
\end{proof}

Now we combine the gathered ingredients to show that a random scheme yields a high probability of a collision with any fixed $q$-subset. At this point we also adjust $d$ to be larger then the factors depending on $k$ and $q$. 

\begin{lemma}\label{lem:scheme:randomBob}
    Let $k,d,q \in \nn$ satisfy $d \ge k \cdot (4q)^{4k\log k}$.
    Consider some $q$-subset $A \sub L(T_{k,d})$.
    Suppose that we choose the scheme $\beta = (f_v)_{v \in V(T_{k,d})} \in \schemes(k,d)$ by picking each bijection $f_v \colon L(T^v_{k,d}) \to [|L(T^v_{k,d})|]$ uniformly and independently at random.
    Then the probability that $(A,\beta)$ has no collision is at most $\exp(-k^d)$.
\end{lemma}
\begin{proof}
%We will write $s(v)$ for $|L(T^{v}_{k,d})|$.
    We apply \Cref{lem:scheme:F} and use the obtained set $F \sub V(T_{k,d})$ to analyze the probability of getting a collision.
    Consider $u \in F$ with $\children(u) = \{u_1,\dots,u_k\}$ and let $C_u$ denote the event that $(A,\beta)$ has a collision at $u$.
    For each $i \in [k]$ we have $\leaves(u_i) = \leaves(u)/k$ and we know from \Cref{lem:scheme:F}(1) that $\fr_A(u_i) \ge 1/(4q)$. %$=  {\leaves(u)}/(4kq)$.
    For each $i \in [k]$ a random bijection $f_{u_i}$ is chosen between $L(T^{u_i}_{k,d})$ and $[\leaves(u_i)]$.
    This can be interpreted as first picking an arbitrary bijection to $[\leaves(u_i)]$ and then combining it with a random permutation over $[\leaves(u_i)]$.
    We apply \Cref{lem:scheme:assosiation} with $z = 4q$ to infer that the probability of getting no collision at $u$ is upper bounded by
    \[
    \prob{\neg C_u} \le \exp\br{\frac {-\leaves(u_i)} {z^k}} = \exp\br{\frac{-\leaves(u)}{k\cdot(4q)^k}}. \]
    Since the sets $(\children(u))_{u \in F}$ are pairwise disjoint, the corresponding events $C_u$ are independent.
    We can thus upper bound the probability of getting no collision at all by the product $\prod_{u\in F} \prob{\neg C_u}$.
    %$\Pi_{u \in F} \exp\br{\frac{-\leaves(u)}{k \cdot(4q)^k)}}$.
    Next, by \Cref{lem:scheme:F}(2) and the assumption on $d$ we know that
    \[ \sum_{u \in F} \leaves(u) \ge d\cdot k^{d} \cdot (4q)^{-3k \log k} \ge k^{d+1} \cdot (4q)^k .\]
    We combine this with the previous formula to obtain
    \[
    \prob{\neg \bigcup_{u\in F} C_u} = \prob{\bigcap_{u\in F} \neg C_u} = \prod_{u\in F} \prob{\neg C_u} \le \exp\br{\frac{-\sum_{u\in F} \leaves(u)}{k\cdot(4q)^k}} \le  \exp(-k^d). % \exp\br{\frac{-d\cdot k^d}{k\cdot\log(q)\cdot(4q)^{k+2}}}.
    \]
\end{proof}

We are ready to prove \Cref{thm:game:bob} (restated below) and thus finish the proof of the reduction.

\thmBob*
\begin{proof}
    We choose the scheme $\beta$ by picking each bijection uniformly and independently at random.
    For a fixed $q$-subset $A$ let $C_A$ denote the event that $(A,\beta)$ witnesses a collision. 
    In these terms, \Cref{lem:scheme:randomBob}
    says that $\prob{\neg C_A} \le \exp(-k^d)$. %$ \le (\frac 1 3)^{k^d}$.
    Let $\mathcal{A}$ be the family of all $q$-subsets $A \sub L(T_{k,d})$; we have $|\mathcal{A}| \le 2^{k^d}$.
    By the union bound, the probability that there exists a $q$-subset with no collision with $\beta$ is 
    \[ \prob{\bigcup_{A \in \mathcal{A}} \neg C_A} \le \sum_{A \in \mathcal{A}} \prob{\neg C_A} \le 2^{k^d} \cdot (1/e)^{k^d} < 1.\]
    Consequently, there is a positive probability of choosing a scheme $\beta$ having a collision with every $q$-subset.
    In particular, this means that such a scheme exists.
\end{proof}

\section{Conclusion}
We have shown that no FPT algorithm can achieve an $\Oh(1)$-approximation for {\sc Max Disjoint Paths} on acyclic digraphs.
However, our reduction blows up the parameter significantly so it does not preserve a running time of the form $f(k)n^{o(k)}$.
It is known that such a running time is unlikely for the exact variant of the problem~\cite{Chitnis23}.
This leads to a question whether {\sc Max} \dagdp admits an $\Oh(1)$-approximation that is faster than $n^{\Oh(k)}$.

Our proof yields an alternative technique for gap amplification in a parameterized reduction based on the probabilistic method (extending the restricted version appearing in~\cite{steiner-orient}), compared to reductions relying on coding theory~\cite{Lin21, guruswami2023parameterized} or communication complexity~\cite{dominating-set}.
Can this approach come in useful for proving that Parameterized Inapproximability Hypothesis (PIH) follows from FPT$\ne$W[1]?

\bibliographystyle{plainurl}
\bibliography{main}

\begin{thebibliography}{10}

\bibitem{AdlerKKLST17}
Isolde Adler, Stavros~G. Kolliopoulos, Philipp~Klaus Krause, Daniel Lokshtanov,
  Saket Saurabh, and Dimitrios~M. Thilikos.
\newblock Irrelevant vertices for the planar disjoint paths problem.
\newblock {\em J. Comb. Theory, Ser. {B}}, 122:815--843, 2017.
\newblock \href {https://doi.org/10.1016/j.jctb.2016.10.001}
  {\path{doi:10.1016/j.jctb.2016.10.001}}.

\bibitem{amiri2016routing}
Saeed~Akhoondian Amiri, Stephan Kreutzer, D\'{a}niel Marx, and Roman
  Rabinovich.
\newblock {Routing with Congestion in Acyclic Digraphs}.
\newblock In Piotr Faliszewski, Anca Muscholl, and Rolf Niedermeier, editors,
  {\em 41st International Symposium on Mathematical Foundations of Computer
  Science (MFCS 2016)}, volume~58 of {\em Leibniz International Proceedings in
  Informatics (LIPIcs)}, pages 7:1--7:11, Dagstuhl, Germany, 2016. Schloss
  Dagstuhl -- Leibniz-Zentrum f{\"u}r Informatik.
\newblock \href {https://doi.org/10.4230/LIPIcs.MFCS.2016.7}
  {\path{doi:10.4230/LIPIcs.MFCS.2016.7}}.

\bibitem{bentert2024tight}
Matthias Bentert, Fedor~V. Fomin, and Petr~A. Golovach.
\newblock Tight approximation and kernelization bounds for vertex-disjoint
  shortest paths.
\newblock abs/2402.15348, 2024.
\newblock \href {http://arxiv.org/abs/2402.15348} {\path{arXiv:2402.15348}}.

\bibitem{bhattacharyya2021parameterized}
Arnab Bhattacharyya, {\'E}douard Bonnet, L{\'a}szl{\'o} Egri, Suprovat Ghoshal,
  Bingkai Lin, Pasin Manurangsi, and D{\'a}niel Marx.
\newblock Parameterized intractability of even set and shortest vector problem.
\newblock {\em Journal of the ACM (JACM)}, 68(3):1--40, 2021.
\newblock \href {https://doi.org/10.1145/3444942} {\path{doi:10.1145/3444942}}.

\bibitem{cavallaro2024edge}
Dario~Giuliano Cavallaro, Ken{-}ichi Kawarabayashi, and Stephan Kreutzer.
\newblock Edge-disjoint paths in eulerian digraphs.
\newblock In Bojan Mohar, Igor Shinkar, and Ryan O'Donnell, editors, {\em
  Proceedings of the 56th Annual {ACM} Symposium on Theory of Computing, {STOC}
  2024, Vancouver, BC, Canada, June 24-28, 2024}, pages 704--715. {ACM}, 2024.
\newblock \href {https://doi.org/10.1145/3618260.3649758}
  {\path{doi:10.1145/3618260.3649758}}.

\bibitem{gap-eth}
Parinya Chalermsook, Marek Cygan, Guy Kortsarz, Bundit Laekhanukit, Pasin
  Manurangsi, Danupon Nanongkai, and Luca Trevisan.
\newblock From {Gap-Exponential Time Hypothesis} to fixed parameter tractable
  inapproximability: Clique, dominating set, and more.
\newblock {\em {SIAM} J. Comput.}, 49(4):772--810, 2020.
\newblock \href {https://doi.org/10.1137/18M1166869}
  {\path{doi:10.1137/18M1166869}}.

\bibitem{chalermsook2014pre}
Parinya Chalermsook, Bundit Laekhanukit, and Danupon Nanongkai.
\newblock Pre-reduction graph products: Hardnesses of properly learning {DFAs}
  and approximating {EDP} on dags.
\newblock In {\em 2014 IEEE 55th Annual Symposium on Foundations of Computer
  Science}, pages 444--453. IEEE, 2014.
\newblock \href {https://doi.org/10.1109/FOCS.2014.54}
  {\path{doi:10.1109/FOCS.2014.54}}.

\bibitem{chekuri2004edge}
Chandra Chekuri, Sanjeev Khanna, and F.~Bruce Shepherd.
\newblock Edge-disjoint paths in planar graphs.
\newblock In {\em 45th Annual IEEE Symposium on Foundations of Computer
  Science}, pages 71--80. IEEE, 2004.
\newblock \href {https://doi.org/10.1109/FOCS.2004.27}
  {\path{doi:10.1109/FOCS.2004.27}}.

\bibitem{chekuri2006}
Chandra Chekuri, Sanjeev Khanna, and F.~Bruce Shepherd.
\newblock An {$O(\sqrt n)$} approximation and integrality gap for disjoint
  paths and unsplittable flow.
\newblock {\em Theory of computing}, 2(1):137--146, 2006.
\newblock \href {https://doi.org/10.4086/TOC.2006.V002A007}
  {\path{doi:10.4086/TOC.2006.V002A007}}.

\bibitem{chekuri2006edge}
Chandra Chekuri, Sanjeev Khanna, and F.~Bruce Shepherd.
\newblock Edge-disjoint paths in planar graphs with constant congestion.
\newblock {\em {SIAM} J. Comput.}, 39(1):281--301, 2009.
\newblock \href {https://doi.org/10.1137/060674442}
  {\path{doi:10.1137/060674442}}.

\bibitem{ChekuriKS09}
Chandra Chekuri, Sanjeev Khanna, and F.~Bruce Shepherd.
\newblock A note on multiflows and treewidth.
\newblock {\em Algorithmica}, 54(3):400--412, 2009.
\newblock \href {https://doi.org/10.1007/S00453-007-9129-Z}
  {\path{doi:10.1007/S00453-007-9129-Z}}.

\bibitem{Chitnis23}
Rajesh Chitnis.
\newblock A tight lower bound for edge-disjoint paths on planar dags.
\newblock {\em SIAM Journal on Discrete Mathematics}, 37(2):556--572, 2023.
\newblock \href {https://doi.org/10.1137/21M1395089}
  {\path{doi:10.1137/21M1395089}}.

\bibitem{Cho0O23}
Kyungjin Cho, Eunjin Oh, and Seunghyeok Oh.
\newblock Parameterized algorithm for the disjoint path problem on planar
  graphs: Exponential in \emph{k}\({}^{\mbox{2}}\) and linear in \emph{n}.
\newblock In Nikhil Bansal and Viswanath Nagarajan, editors, {\em Proceedings
  of the 2023 {ACM-SIAM} Symposium on Discrete Algorithms, {SODA} 2023,
  Florence, Italy, January 22-25, 2023}, pages 3734--3758. {SIAM}, 2023.
\newblock \href {https://doi.org/10.1137/1.9781611977554.CH144}
  {\path{doi:10.1137/1.9781611977554.CH144}}.

\bibitem{chuzhoy2017new}
Julia Chuzhoy, David H.~K. Kim, and Rachit Nimavat.
\newblock New hardness results for routing on disjoint paths.
\newblock {\em {SIAM} J. Comput.}, 51(2):17--189, 2022.
\newblock \href {https://doi.org/10.1137/17M1146580}
  {\path{doi:10.1137/17M1146580}}.

\bibitem{chuzhoy2016improved}
Julia Chuzhoy, David~HK Kim, and Shi Li.
\newblock Improved approximation for node-disjoint paths in planar graphs.
\newblock In {\em Proceedings of the forty-eighth annual ACM symposium on
  Theory of Computing}, pages 556--569, 2016.
\newblock \href {https://doi.org/10.1145/2897518.2897538}
  {\path{doi:10.1145/2897518.2897538}}.

\bibitem{cygan2015parameterized}
Marek Cygan, Fedor~V. Fomin, Lukasz Kowalik, Daniel Lokshtanov, D{\'{a}}niel
  Marx, Marcin Pilipczuk, Michal Pilipczuk, and Saket Saurabh.
\newblock {\em Parameterized Algorithms}.
\newblock Springer, 2015.
\newblock \href {https://doi.org/10.1007/978-3-319-21275-3}
  {\path{doi:10.1007/978-3-319-21275-3}}.

\bibitem{DBLP:conf/focs/CyganMPP13}
Marek Cygan, D{\'{a}}niel Marx, Marcin Pilipczuk, and Michal Pilipczuk.
\newblock The planar directed k-vertex-disjoint paths problem is
  fixed-parameter tractable.
\newblock In {\em 54th Annual {IEEE} Symposium on Foundations of Computer
  Science, {FOCS} 2013, 26-29 October, 2013, Berkeley, CA, {USA}}, pages
  197--206, 2013.
\newblock \href {https://doi.org/10.1109/FOCS.2013.29}
  {\path{doi:10.1109/FOCS.2013.29}}.

\bibitem{diestel-book}
Reinhard Diestel.
\newblock {\em Graph Theory, 4th Edition}, volume 173 of {\em Graduate texts in
  mathematics}.
\newblock Springer, 2012.

\bibitem{Dinur07}
Irit Dinur.
\newblock The {PCP} theorem by gap amplification.
\newblock {\em J. {ACM}}, 54(3):12, 2007.
\newblock \href {https://doi.org/10.1145/1236457.1236459}
  {\path{doi:10.1145/1236457.1236459}}.

\bibitem{EneMPR16}
Alina Ene, Matthias Mnich, Marcin Pilipczuk, and Andrej Risteski.
\newblock On routing disjoint paths in bounded treewidth graphs.
\newblock In Rasmus Pagh, editor, {\em 15th Scandinavian Symposium and
  Workshops on Algorithm Theory, {SWAT} 2016, June 22-24, 2016, Reykjavik,
  Iceland}, volume~53 of {\em LIPIcs}, pages 15:1--15:15. Schloss Dagstuhl -
  Leibniz-Zentrum f{\"{u}}r Informatik, 2016.
\newblock \href {https://doi.org/10.4230/LIPICS.SWAT.2016.15}
  {\path{doi:10.4230/LIPICS.SWAT.2016.15}}.

\bibitem{feldmann2020survey}
Andreas~Emil Feldmann, Karthik~C S, Euiwoong Lee, and Pasin Manurangsi.
\newblock A survey on approximation in parameterized complexity: Hardness and
  algorithms.
\newblock {\em Algorithms}, 13(6):146, 2020.
\newblock \href {https://doi.org/10.3390/A13060146}
  {\path{doi:10.3390/A13060146}}.

\bibitem{fortune1980directed}
Steven Fortune, John~E. Hopcroft, and James Wyllie.
\newblock The directed subgraph homeomorphism problem.
\newblock {\em Theor. Comput. Sci.}, 10:111--121, 1980.
\newblock \href {https://doi.org/10.1016/0304-3975(80)90009-2}
  {\path{doi:10.1016/0304-3975(80)90009-2}}.

\bibitem{frank1990packing}
Andr{\'a}s Frank.
\newblock Packing paths, cuts, and circuits - a survey.
\newblock {\em Paths, Flows and VLSI-Layout}, 49:100, 1990.

\bibitem{GiannopoulouKKK22}
Archontia~C. Giannopoulou, Ken{-}ichi Kawarabayashi, Stephan Kreutzer, and
  O{-}joung Kwon.
\newblock Directed tangle tree-decompositions and applications.
\newblock In Joseph~(Seffi) Naor and Niv Buchbinder, editors, {\em Proceedings
  of the 2022 {ACM-SIAM} Symposium on Discrete Algorithms, {SODA} 2022, Virtual
  Conference / Alexandria, VA, USA, January 9 - 12, 2022}, pages 377--405.
  {SIAM}, 2022.
\newblock \href {https://doi.org/10.1137/1.9781611977073.19}
  {\path{doi:10.1137/1.9781611977073.19}}.

\bibitem{guruswami2023parameterized}
Venkatesan Guruswami, Bingkai Lin, Xuandi Ren, Yican Sun, and Kewen Wu.
\newblock Parameterized inapproximability hypothesis under exponential time
  hypothesis.
\newblock In Bojan Mohar, Igor Shinkar, and Ryan O'Donnell, editors, {\em
  Proceedings of the 56th Annual {ACM} Symposium on Theory of Computing, {STOC}
  2024, Vancouver, BC, Canada, June 24-28, 2024}, pages 24--35. {ACM}, 2024.
\newblock \href {https://doi.org/10.1145/3618260.3649771}
  {\path{doi:10.1145/3618260.3649771}}.

\bibitem{GuruswamiRS23}
Venkatesan Guruswami, Xuandi Ren, and Sai Sandeep.
\newblock {Baby PIH: Parameterized Inapproximability of Min CSP}.
\newblock In Rahul Santhanam, editor, {\em 39th Computational Complexity
  Conference (CCC 2024)}, volume 300 of {\em Leibniz International Proceedings
  in Informatics (LIPIcs)}, pages 27:1--27:17, Dagstuhl, Germany, 2024. Schloss
  Dagstuhl -- Leibniz-Zentrum f{\"u}r Informatik.
\newblock \href {https://doi.org/10.4230/LIPIcs.CCC.2024.27}
  {\path{doi:10.4230/LIPIcs.CCC.2024.27}}.

\bibitem{heggernes2015finding}
Pinar Heggernes, Pim van’t Hof, Erik~Jan van Leeuwen, and Reza Saei.
\newblock Finding disjoint paths in split graphs.
\newblock {\em Theory of Computing Systems}, 57:140--159, 2015.
\newblock \href {https://doi.org/10.1007/S00224-014-9580-6}
  {\path{doi:10.1007/S00224-014-9580-6}}.

\bibitem{joag1983negative}
Kumar Joag-Dev and Frank Proschan.
\newblock Negative association of random variables with applications.
\newblock {\em The Annals of Statistics}, pages 286--295, 1983.

\bibitem{SK22}
{Karthik {C. S.}} and Subhash Khot.
\newblock Almost polynomial factor inapproximability for parameterized
  k-{Clique}.
\newblock In Shachar Lovett, editor, {\em 37th Computational Complexity
  Conference, {CCC} 2022, July 20-23, 2022, Philadelphia, PA, {USA}}, volume
  234 of {\em LIPIcs}, pages 6:1--6:21. Schloss Dagstuhl - Leibniz-Zentrum
  f{\"{u}}r Informatik, 2022.
\newblock \href {https://doi.org/10.4230/LIPICS.CCC.2022.6}
  {\path{doi:10.4230/LIPICS.CCC.2022.6}}.

\bibitem{dominating-set}
{Karthik {C. S.}}, Bundit Laekhanukit, and Pasin Manurangsi.
\newblock On the parameterized complexity of approximating dominating set.
\newblock {\em J. {ACM}}, 66(5):33:1--33:38, 2019.
\newblock \href {https://doi.org/10.1145/3325116} {\path{doi:10.1145/3325116}}.

\bibitem{kawarabayashi2014excluded}
Ken-ichi Kawarabayashi, Yusuke Kobayashi, and Stephan Kreutzer.
\newblock An excluded half-integral grid theorem for digraphs and the directed
  disjoint paths problem.
\newblock In {\em Proceedings of the Forty-Sixth Annual ACM Symposium on Theory
  of Computing}, STOC '14, page 70–78, New York, NY, USA, 2014. Association
  for Computing Machinery.
\newblock \href {https://doi.org/10.1145/2591796.2591876}
  {\path{doi:10.1145/2591796.2591876}}.

\bibitem{kawarabayashi2012disjoint}
Ken-ichi Kawarabayashi, Yusuke Kobayashi, and Bruce Reed.
\newblock The disjoint paths problem in quadratic time.
\newblock {\em Journal of Combinatorial Theory, Series B}, 102(2):424--435,
  2012.
\newblock \href {https://doi.org/10.1016/J.JCTB.2011.07.004}
  {\path{doi:10.1016/J.JCTB.2011.07.004}}.

\bibitem{kawarabayashi2015grid}
Ken-ichi Kawarabayashi and Stephan Kreutzer.
\newblock The directed grid theorem.
\newblock In {\em Proceedings of the Forty-Seventh Annual ACM Symposium on
  Theory of Computing}, STOC '15, page 655–664, New York, NY, USA, 2015.
  Association for Computing Machinery.
\newblock \href {https://doi.org/10.1145/2746539.2746586}
  {\path{doi:10.1145/2746539.2746586}}.

\bibitem{kleinberg1995approximations}
Jon~M. Kleinberg and {\'{E}}va Tardos.
\newblock Approximations for the disjoint paths problem in high-diameter planar
  networks.
\newblock {\em J. Comput. Syst. Sci.}, 57(1):61--73, 1998.
\newblock \href {https://doi.org/10.1006/JCSS.1998.1579}
  {\path{doi:10.1006/JCSS.1998.1579}}.

\bibitem{kolliopoulos2004approximating}
Stavros~G Kolliopoulos and Clifford Stein.
\newblock Approximating disjoint-path problems using packing integer programs.
\newblock {\em Mathematical Programming}, 99(1):63--87, 2004.
\newblock \href {https://doi.org/10.1007/S10107-002-0370-6}
  {\path{doi:10.1007/S10107-002-0370-6}}.

\bibitem{korhonen2024minor}
Tuukka Korhonen, Michał Pilipczuk, and Giannos Stamoulis.
\newblock Minor containment and disjoint paths in almost-linear time.
\newblock abs/2404.03958, 2024 (to appear at FOCS 2024).
\newblock \href {http://arxiv.org/abs/2404.03958} {\path{arXiv:2404.03958}}.

\bibitem{kramer1984complexity}
Mark~R. Kramer and Jan van Leeuwen.
\newblock The complexity of wire-routing and finding minimum area layouts for
  arbitrary {VLSI} circuits.
\newblock {\em Advances in Computing Research}, 2:129--146, 1984.

\bibitem{lampis2024parameterized}
Michael Lampis and Manolis Vasilakis.
\newblock Parameterized maximum node-disjoint paths.
\newblock abs/2404.14849, 2024.
\newblock \href {http://arxiv.org/abs/2404.14849} {\path{arXiv:2404.14849}}.

\bibitem{Lin21}
Bingkai Lin.
\newblock Constant approximating k-{Clique} is {W[1]}-hard.
\newblock In Samir Khuller and Virginia~Vassilevska Williams, editors, {\em
  {STOC} '21: 53rd Annual {ACM} {SIGACT} Symposium on Theory of Computing,
  Virtual Event, Italy, June 21-25, 2021}, pages 1749--1756. {ACM}, 2021.
\newblock \href {https://doi.org/10.1145/3406325.3451016}
  {\path{doi:10.1145/3406325.3451016}}.

\bibitem{LinRSW23}
Bingkai Lin, Xuandi Ren, Yican Sun, and Xiuhan Wang.
\newblock Constant approximating parameterized $k$-{SetCover is W[2]-hard}.
\newblock In Nikhil Bansal and Viswanath Nagarajan, editors, {\em Proceedings
  of the 2023 {ACM-SIAM} Symposium on Discrete Algorithms, {SODA} 2023,
  Florence, Italy, January 22-25, 2023}, pages 3305--3316. {SIAM}, 2023.
\newblock \href {https://doi.org/10.1137/1.9781611977554.CH126}
  {\path{doi:10.1137/1.9781611977554.CH126}}.

\bibitem{LokshtanovMPSZ20}
Daniel Lokshtanov, Pranabendu Misra, Micha\l{} Pilipczuk, Saket Saurabh, and
  Meirav Zehavi.
\newblock An exponential time parameterized algorithm for planar disjoint
  paths.
\newblock In {\em Proceedings of the 52nd Annual ACM SIGACT Symposium on Theory
  of Computing}, STOC 2020, page 1307–1316, New York, NY, USA, 2020.
  Association for Computing Machinery.
\newblock \href {https://doi.org/10.1145/3357713.3384250}
  {\path{doi:10.1145/3357713.3384250}}.

\bibitem{doct}
Daniel Lokshtanov, M.~S. Ramanujan, Saket Saurabh, and Meirav Zehavi.
\newblock Parameterized complexity and approximability of directed odd cycle
  transversal.
\newblock In {\em Proceedings of the 2020 {ACM-SIAM} Symposium on Discrete
  Algorithms, {SODA} 2020, Salt Lake City, UT, USA, January 5-8, 2020}, pages
  2181--2200, 2020.
\newblock \href {https://doi.org/10.1137/1.9781611975994.134}
  {\path{doi:10.1137/1.9781611975994.134}}.

\bibitem{lokshtanov2020parameterized}
Daniel Lokshtanov, MS~Ramanujan, Saket Saurab, and Meirav Zehavi.
\newblock Parameterized complexity and approximability of directed odd cycle
  transversal.
\newblock In {\em Proceedings of the Fourteenth Annual ACM-SIAM Symposium on
  Discrete Algorithms}, pages 2181--2200. SIAM, 2020.
\newblock \href {https://doi.org/10.1137/1.9781611975994.134}
  {\path{doi:10.1137/1.9781611975994.134}}.

\bibitem{Lokshtanov0S20}
Daniel Lokshtanov, Saket Saurabh, and Vaishali Surianarayanan.
\newblock A parameterized approximation scheme for min $k$-{Cut}.
\newblock In Sandy Irani, editor, {\em 61st {IEEE} Annual Symposium on
  Foundations of Computer Science, {FOCS} 2020, Durham, NC, USA, November
  16-19, 2020}, pages 798--809. {IEEE}, 2020.
\newblock \href {https://doi.org/10.1109/FOCS46700.2020.00079}
  {\path{doi:10.1109/FOCS46700.2020.00079}}.

\bibitem{lynch1975equivalence}
James~F Lynch.
\newblock The equivalence of theorem proving and the interconnection problem.
\newblock {\em ACM SIGDA Newsletter}, 5(3):31--36, 1975.

\bibitem{natarajan1996disjoint}
Sridhar Natarajan and Alan~P Sprague.
\newblock Disjoint paths in circular arc graphs.
\newblock {\em Nordic Journal of Computing}, 3(3):256--270, 1996.

\bibitem{ogier1993distributed}
Richard~G. Ogier, Vladislav Rutenburg, and Nachum Shacham.
\newblock Distributed algorithms for computing shortest pairs of disjoint
  paths.
\newblock {\em {IEEE} Trans. Inf. Theory}, 39(2):443--455, 1993.
\newblock \href {https://doi.org/10.1109/18.212275}
  {\path{doi:10.1109/18.212275}}.

\bibitem{ohsaka2024parameterized}
Naoto Ohsaka.
\newblock On the parameterized intractability of determinant maximization.
\newblock {\em Algorithmica}, pages 1--33, 2024.
\newblock \href {https://doi.org/10.1007/S00453-023-01205-0}
  {\path{doi:10.1007/S00453-023-01205-0}}.

\bibitem{reed1995rooted}
Bruce Reed.
\newblock Rooted routing in the plane.
\newblock {\em Discrete Applied Mathematics}, 57(2-3):213--227, 1995.
\newblock \href {https://doi.org/10.1016/0166-218X(94)00104-L}
  {\path{doi:10.1016/0166-218X(94)00104-L}}.

\bibitem{robertson1995graph}
Neil Robertson and Paul~D Seymour.
\newblock Graph minors. {XIII. The} disjoint paths problem.
\newblock {\em Journal of combinatorial theory, Series B}, 63(1):65--110, 1995.
\newblock \href {https://doi.org/10.1006/jctb.1995.1006}
  {\path{doi:10.1006/jctb.1995.1006}}.

\bibitem{schrijver2003combinatorial}
Alexander Schrijver.
\newblock {\em Combinatorial optimization: polyhedra and efficiency},
  volume~24.
\newblock Springer, 2003.

\bibitem{Slivkins10}
Aleksandrs Slivkins.
\newblock Parameterized tractability of edge-disjoint paths on directed acyclic
  graphs.
\newblock {\em {SIAM} J. Discret. Math.}, 24(1):146--157, 2010.
\newblock \href {https://doi.org/10.1137/070697781}
  {\path{doi:10.1137/070697781}}.

\bibitem{srinivas2005finding}
Anand Srinivas and Eytan~H. Modiano.
\newblock Finding minimum energy disjoint paths in wireless ad-hoc networks.
\newblock {\em Wirel. Networks}, 11(4):401--417, 2005.
\newblock \href {https://doi.org/10.1007/S11276-005-1765-0}
  {\path{doi:10.1007/S11276-005-1765-0}}.

\bibitem{steiner-orient}
Micha\l{} W\l{}odarczyk.
\newblock Parameterized inapproximability for steiner orientation by gap
  amplification.
\newblock In Artur Czumaj, Anuj Dawar, and Emanuela Merelli, editors, {\em 47th
  International Colloquium on Automata, Languages, and Programming, {ICALP}
  2020, July 8-11, 2020, Saarbr{\"{u}}cken, Germany (Virtual Conference)},
  volume 168 of {\em LIPIcs}, pages 104:1--104:19. Schloss Dagstuhl -
  Leibniz-Zentrum f{\"{u}}r Informatik, 2020.
\newblock \href {https://doi.org/10.4230/LIPICS.ICALP.2020.104}
  {\path{doi:10.4230/LIPICS.ICALP.2020.104}}.

\bibitem{WlodarczykZ23}
Micha\l{} W\l{}odarczyk and Meirav Zehavi.
\newblock Planar disjoint paths, treewidth, and kernels.
\newblock In {\em 64th {IEEE} Annual Symposium on Foundations of Computer
  Science, {FOCS} 2023, Santa Cruz, CA, USA, November 6-9, 2023}, pages
  649--662. {IEEE}, 2023.
\newblock \href {https://doi.org/10.1109/FOCS57990.2023.00044}
  {\path{doi:10.1109/FOCS57990.2023.00044}}.

\end{thebibliography}

\end{document}